\documentclass[11pt,a4paper]{article}
\usepackage{times}
\usepackage{hyperref}
\usepackage{graphicx}
\usepackage{amssymb, amsmath, amsthm}
\usepackage{framed}
\usepackage{dsfont}\usepackage{enumerate}\usepackage{color}

\newtheorem{Remark}{Remark}
\newtheorem{Definition}{Definition}
\newtheorem{Theorem}{Theorem}
\newtheorem{Proposition}{Proposition}
\newtheorem{Lemma}{Lemma}

\newcommand{\R}{{\mathbb R}}
\newcommand{\ff}{{\mathbb F}}
\newcommand{\bG}{{\mathbb G}}

\newcommand{\F}{{\mathcal F}}
\newcommand{\G}{{\mathcal G}}
\newcommand{\cL}{{\mathcal L}}

\newcommand{\E}{{\mathcal E}}
\newcommand{\A}{{\mathcal A}}
\newcommand{\M}{{\mathcal M}}

\newcommand{\Ss}{{\mathcal S}}
\newcommand{\Hh}{{\mathcal H}}

\def\reff#1{{\rm(\ref{#1})}}
\def \bQ{\mathbb{Q}}
\def \bP{\mathbb{P}}
\def \bE{\mathbb{E}}

\def \cW{\mathcal{W}}

\def \1{\mathds{1}}

\def \wt {\widetilde}
\def \wh {\widehat}

\newcommand{\nc}{\newcommand}
\nc{\esssup}{\mathop{\mathrm{ess\;sup}}}
\nc{\essinf}{\mathop{\mathrm{ess\;inf}}}

\begin{document}
\title{Optimization problem under change of regime of interest rate\footnote{The research of the first author
was supported  by the Sectorial Operational Programme Human
Resources Development (SOP HRD), financed from the European Social
Fund and by the Romanian Government under the contract number SOP
HRD/89/1.5/S/62988. The research of the three
 other authors is
supported by Chaire risque de cr\'edit, French Banking
Federation}}

\author{Bogdan Iftimie\\ \footnotesize{Bucharest University of   Economic Studies}, \\ \footnotesize{and ``Simion Stoilow" Institute of Romanian Academy} \\
\footnotesize{ \texttt{Bogdan.Iftimie@csie.ase.ro}} 
\and Monique Jeanblanc \\ \footnotesize{Laboratoire d'Analyse et Probabilit\'es,}\\ \footnotesize{Universit\'e d'Evry-Val d'Essonne}\\ \footnotesize{\texttt{monique.jeanblanc@univ-evry.fr}}
 \and Thomas Lim\\ \footnotesize{Laboratoire d'Analyse et Probabilit\'es,}\\ \footnotesize{Universit\'e d'Evry-Val d'Essonne}\footnotesize{ and ENSIIE}\\ \footnotesize{\texttt{lim@ensiie.fr}}
\and Hai-Nam Nguyen  \\ \footnotesize{Laboratoire d'Analyse et Probabilit\'es,}\\ \footnotesize{Universit\'e d'Evry-Val d'Essonne}  \\\footnotesize{\texttt{hainam.nguyen@univ-evry.fr}}  
}
\date{}\maketitle
\abstract{In this paper, we study the classical problem of maximization of the sum of the utility of the terminal wealth and the utility of the consumption, in a case where a sudden jump in the risk-free interest rate creates incompleteness. The value function  of the  dual problem is  proved to be solution of a  BSDE and the duality between the primal and the dual value  functions is exploited to study the BSDE associated to the primal problem.}
\vspace*{1cm}



\section{Introduction}
 Many studies  in the field of Mathematical Finance   are devoted to    portfolio   and/or consumption optimization problems.
In the case of a complete market, with several risky assets and a savings account adapted to a Brownian filtration, the problem is fully solved in the monography  of Karatzas and Shreve \cite{KSDual}.
The situation in incomplete markets is more delicate, and it is not easy to give closed
 form solutions (see, e.g., Menoncin \cite{M}). The incompleteness of the market may  arise from a  number of risky assets smaller than the dimension of the driving noise,  from constraints on the portfolio, or from an interest rate which  depends on an extra noise, which will be the case in our setting. The literature about the two first cases of incompleteness is important, on the other hand the literature about the third case of incompleteness is reduced.
We can cite \cite{M} for the case of a multidimensional incomplete market (and constant interest rate) and a Brownian filtration under Markovian framework, where the author solves the problem using HJB equation. The case where the measurability of the interest rate creates incompleteness is presented in    Bauerle and   Rieder  \cite{BR} in which the dynamics of the interest rate is driven by a Markov chain.
 %
%  Hu, Imkeller and M\"{u}ller \cite{HIM} solve a portfolio optimization problem, with a constraint   on trading strategies, for logarithmic, power  and exponential utilities, using the  BSDEs  approach. The BSDE approach is  also the main tool in
%Jiao and Pham (\cite{JP}) who are dealing  with a portfolio optimization problem  in a  financial market consisting in a savings account (with constant  interest rate) and a stock driven by a Brownian motion which  suffer also  a finite number of unpredicted  shock  (implying the incompleteness of the market)  %under the  density hypothesis,
%for HARA utility function. In the case of a single jump, the original problem is decomposed in two subproblems: before the jump and after the jump, which are stated in complete markets, the latter being solved by martingale duality and the former by dynamic programming approach.
%
%
%

 %%The case of stochastic  interest rates is presented in \cite{KL}, Korn .. (compl\`eter)  in the case where this does not create incompleteness

A classical tool to solve utility maximization problem is the dual approach. This one consists in solving  an auxiliary optimization problem, called the dual problem, which is defined on the set of all equivalent martingale measures. The list of papers studying that problem is long and we quote only few of them.  This  approach is used in the case of incomplete markets generated by a savings account (with constant interest rate) and several stocks (represented by general semi-martingales) for HARA utility,   by  Kramkov and Schachermayer \cite{KS}. They  state an existence and uniqueness result for the final optimal wealth (associated to an investment problem), but no explicit formulas are provided.  Rogers \cite{Ro} formulates an abstract theorem in which the value function of the  utility maximization problem and   the value function for the associated dual problem  satisfy a bidual relation.  As it is mentioned, this procedure  can be applied for a wide class of portfolio  and/or consumption optimization problems.   Casta\~{n}eda-Leyva and Hern\'{a}ndez-Hern\'{a}ndez \cite{CH}  deal with a combined investment and consumption optimization problem with a single risky asset, in a Brownian framework, and where the  coefficients of the model (including the interest rate) are deterministic functions  of some external economic factor process.\\

%The dual problem can be constructed and the link between the value functions of the primal, respectively the dual can be established by means of general abstract results.

Here, we are  concerned with  the problem of maximization of expected power utility of both terminal wealth and   consumption, in a market with investment opportunities in a savings account with  a  stochastic interest rate,  which  suffers an unexpected shock at some random time $\tau$, and a stock modeled by  a semi-martingale driven by a Brownian motion. The unexpected shock can for example be due to some serious macroeconomic issue. This one implies that the market is incomplete. The problem will be solved in the   filtration generated by prices (of stock and savings account)   so that the change of regime time  $\tau$ is a stopping time,
under the immersion hypothesis between  the filtration generated by the stock and the general filtration.

Using standard results of duality, the original optimization problem (called the {\it primal} problem) is linked to the {\it dual} problem, in which the control parameters  take value in the set of equivalent martingale measures. Then, we prove, by using a similar approach to the one used in Hu \emph{et al.} \cite{HIM} for the case of the primal problem without consumption and more recently in Cheridito and Hu \cite{Cheridito} for the case with consumption, that the value function of that problem is solution of a particular BSDE, involving one  jump. Using a recent result of  Kharroubi and Lim  \cite{KL},  we show that this BSDE has a unique solution.
Then, we give the optimal portfolio and consumption in terms of the solution of this BSDE, and    explicit formula   for the optimal wealth process.  We also establish a duality result for the dynamic versions of the value functions associated to primal and dual optimization problems  which allows us to prove that the BSDE associated to the primal problem has a unique solution. To the best of our knowledge, the BSDE methodology has not  been used yet for dual problems in the literature.\\

The paper is organized as follows. In Sections \ref{set up} and \ref{market}, we describe the set up and model. In Section \ref{probleme dual}, we characterize the set of the equivalent martingale measures, then we derive and solve the dual optimization problem. Finally, Section \ref{retour primal} is dedicated to the link between the value functions associated to the primal and dual optimization problems and to the  computation of  explicit formulas for the optimal wealth process, optimal trading and consumption policies.

\section{Set up} \label{set up}
Throughout this paper $(\Omega, \mathcal{G}, \bP)$ %$(\Omega, \cal{A}, P) P \M(P)$
is a   probability space on which is defined a  one dimensional
Brownian motion $(B_t)_{t\in[0, T]}$ where $T < \infty$ is the
terminal time. We denote by $\ff:=(\F_t)_{t\in[0, T]}$    the
natural filtration of  $B$ (augmented by the $\bP$-null sets)  and we
assume that $ \F_T
 \varsubsetneq {\cal {G}}$. On the same probability space is given a finite
positive $\cal G$-measurable random variable $\tau$ which is
interpreted as a random time associated to some unpredicted
evolution (with respect to the filtration $\ff$) in the dynamics of the interest rate or to a switching
regime. Let $H$ be the c\`adl\`ag process equal to $0$ before $\tau$ and $1$ after $\tau$, i.e., $H_t :=
\mathds{1}_{\tau \leq t}$. We introduce the filtration $\bG$ which is the smallest right-continuous extension of $\ff$ that makes $\tau$   a $\bG$-stopping
time. More precisely $\bG:=(\G_t)_{t\in[0, T]}$, where $\G_t$ is defined for any $t\in [0,T]$ by
$$\G_t :=  \bigcap_{\epsilon>0}   \G^0_{t+\epsilon} \;,$$
 where $\G^0_{t} := \F_t\vee
\sigma(H_u\;,u\in[0,t])$, for any $t\in[0, T]$.
Throughout the sequel, we assume the following classical hypotheses.%Our assumptions about the model are the following
\begin{itemize}

\item[{\bf (H1)}]  Any $\ff$-martingale is a $\bG$-martingale,
i.e., $\ff$ is immersed in $\bG$.

\item[{\bf (H2)}] The process $H$ admits an
absolutely continuous  compensator, i.e., there exists a
non-negative $\bG$-adapted process $\lambda^\bG$, called the
$\bG$-intensity, such that the compensated process $M$ defined by
$$M_t := H_t - \int_0^t \lambda_s^\bG ds\;,$$
 is a $\bG$-martingale.
Note that the process $\lambda^\bG$ vanishes after $\tau$, and we can write
$\lambda _t^\bG= \lambda ^\ff_t\mathds{1}_{t < \tau}$ where
$\lambda ^\ff$ is an $\ff$-adapted process,   called the $\ff$-intensity of the process $H$. We assume that $\lambda^\bG$ is uniformly bounded, hence
$\lambda^\ff$ is also uniformly bounded. The existence of $\lambda^\bG$
implies that $\tau$ is not an $\ff$-stopping time (in fact, $\tau$
avoids $\ff$-stopping times and is a totally inaccessible $\bG$-stopping time). %\nhn{$\lambda^\bG, \lambda^\ff, \lambda, \lambda^0$???}
\end{itemize}
We recall  in this framework the standard
decomposition of any $\bG$-predictable process $ \psi $ which is given by    Jeulin \cite[Lemma 4.4]{Jeulin}.
\begin{Lemma}
Any $\bG$-predictable process $ \psi $ can be decomposed under the following form
\begin{equation*}
\psi_t = \psi^0_t \mathds{1}_{t \leq \tau} + \psi^1_t(\tau)
\mathds{1}_{ t > \tau }\;,
\end{equation*}
where the process $\psi^0$ is $\ff$-predictable, and for fixed
non-negative $u$, the process $\psi^1_{\cdot}(u)$
 %(indexed over $u$)
 is  $\ff$-predictable. Furthermore, for any fixed $t \in [0,T]$, the
mapping $\psi^1_t(\cdot)$ is
 $\F_t \otimes \mathcal{B}([0, \infty))$-measurable. Moreover, if the process $\psi$ is uniformly bounded, then it is possible to choose bounded processes $\psi^0$ and $\psi^1(u)$.
\end{Lemma}

\begin{Remark}The process $(\exp (aB_t-\frac 12 a^2t))_{t\in[0, T]}$ being an
$\ff$-continuous martingale for every real number $a$, the immersion property implies that it is a
$\bG$-continuous martingale,  hence $B $ is a
$\bG $-Brownian motion. It follows that the stochastic integral $\int \vartheta_sdB_s$ is well defined for a $\bG$-adapted process $\vartheta$ (up to integrability conditions, e.g. if $\vartheta$ is bounded) and that this integral is a $\bG$ local-martingale.
\end{Remark}

We define the following   spaces which will be used throughout this paper.

\begin{itemize}
\item
$\Ss^\infty_{\ff}(u,T)$ (resp. $\Ss^\infty_{\bG}(u,T)$) denotes the
 set of $\ff$ (resp. $ \bG $)-progressively measurable processes
$X$ which are essentially bounded on $[u,T]$,
i.e., such that
$$
\esssup_{t \in [u, T]} |X_t| <\infty \;;
$$

\item
$\Ss^{\infty, +}_{\ff}(u,T)$ (resp. $\Ss^{\infty, +}_{\bG}(u,T)$) denotes the subset of $\Ss^\infty_{\ff}(u,T)$ (resp. $\Ss^\infty_{\bG}(u,T)$)
  such that  $X_t\geq C$ for a positive constant $C$;

\item $\Hh^2_{\ff}(u,T)$ (resp. $\Hh^2_{\bG}(u,T)$) denotes the
set of square integrable $\ff$ (resp. $\bG$)-predictable processes
$X$ on $[u,T]$, i.e.,
$$ \| X \|^2_{\Hh^2(u,T)} :=  \bE \Big(   \int_u^T \vert X_t\vert^2 dt \Big)    <
\infty \;;
$$%
%$$(resp. \quad \| X
%\|^2_{\Hh^2_{\bG}([0,T])} := E \left(   \int_0^T X_t^2 dt \right)    <
%\infty;)
%$$

\item $\Hh^2_{\bG}(M)$ denotes the set {of   $\bG$-predictable
processes $X $ on $[0, T]$ such that}
$$
\| X \|^2_{\Hh^2_{\bG}(M)} := \bE \Big(   \int_0^T
\lambda_t^\bG\vert X_t\vert^2dt     \Big)    < \infty \;.
$$
\end{itemize}

\section{Model}\label{market}
The financial market consists in a savings account with a stochastic interest rate with dynamics
\begin{equation*}
d S^0_t = r_t S^0_t dt\;, \quad S^0_0 =1\;,
\end{equation*}
where $r$ is a non-negative $\bG$-adapted process, and a risky asset whose price process $S$ follows the dynamics
\begin{equation*}
d S_t = S_t (\nu_t dt  + \sigma_t d B_t) \;.
\end{equation*}
Our assumptions about the market are the following
\begin{itemize}
\item[{\bf (H3)}] $r$ is a $\bG$-adapted process of the form
$$
r_t=r^0_t \mathds{1}_{t < \tau} + r^1_t(\tau)\mathds{1}_{t \geq \tau} \;,
$$
where $r^0$   is a non-negative uniformly bounded $\ff$-adapted
process, and for any fixed non-negative $u$, 
$r^1_{\cdot}(u)$ is a non-negative uniformly bounded $\ff$-adapted process,  and for
fixed $t\in[0, T]$, the mapping $r^1_t(\cdot)$ is $\F_t \otimes \mathcal{B}([0,
\infty))$-measurable.

\item[{\bf (H4)}] $\nu$ and $\sigma$ are $\ff$-adapted processes,
and there exists a positive constant $C$ such that
$|\nu_t| \leq C$ and  $\frac{1}{C} \leq \sigma_t \leq C$,  $
t \in [0, T]$, $\bP$ - a.s.
\end{itemize}
Throughout the sequel, we use the notation $R$ for the discount factor defined by $R_t  := e ^{-\int_0^t r_s ds}$ for any $t\in [0,T]$.\\ 

We now consider an investor acting in this market, starting with an initial amount $x > 0$ and we denote by $\pi^0$ and $\pi$ the part of wealth invested in the savings account and in the risky asset, and by $c$ the associated instantaneous consumption process. Obviously we have the relation $\pi^0_t = 1 - \pi_t$. We denote by $X^{x,\pi,c}$ the wealth process associated to the strategy $(\pi, c)$ and the initial wealth $x$, and we assume that the strategy is self-financing, %\textcolor{red}{(?)}
which leads to the equation
\begin{equation}\label{wealtheq}
\left\{\begin{array}{rcl}
X_0^{x,\pi,c} &= & x\;, \\
d X_t^{x, \pi,c} &=& X_t^{x,\pi,c}  \big[ (r_t  + \pi_t (\nu_t - r_t) ) dt +  \pi_t \sigma_t d B_t \big] -c_tdt\;.
\end{array}\right.
\end{equation}  %Throughout the sequel,
We consider   the set $\A(x)$ of
the admissible strategies defined below.
\begin{Definition}\label{DefinitionAdmissibleStrategies}
The set $\A(x)$ of admissible strategies $(\pi,c)$ consists in $\bG$-predictable processes $(\pi,c)$ such that  $\bE (\int_0^T
|\pi_s\sigma_s|^2 ds)<\infty$, $c_t \geq 0$ and $X^{x,\pi,c}_t >0$ for any $t \in [0,T]$.
\end{Definition}

We are interested in solving the classical problem of utility maximization defined by
\begin{equation}\label{pb}
V(x) := \sup_{(\pi,c) \in \A(x)} \bE \Big[ \int_0^T U(c_s) ds + U(X^{x, \pi, c}_T) \Big] \;,
\end{equation}
where the utility function $U$ is $U(x)=x^p/p$ with $p \in (0,1)$. %$V_0$ is the value function of this optimisation problem.
%If $\pi$ is an admissible strategy,  we set
%$$
%J^{\pi}(x) = E \left[ U \left(X_T^{x, \pi} \right) \right],
%$$
%and the value function of the  optimization problem  is
%$$
%V(x) = \sup_{\pi \in \A} J^{\pi}(x).
%$$

\section{Dual approach} \label{probleme dual}
To prove that there exists an optimal strategy to the problem
(\ref{pb}), we use the dual approach introduced by Karatzas \emph{et al.} \cite{KLS} or Cox and Huang \cite{CH89}.%Kramkov and
%Schachermayer \cite{KS}.% \mb{sans doute aussi El K. Quenez}

For that, we introduce the convex conjugate function $ \wt U $ of
the utility function $U$, which is defined by
$$
 \wt U  (y) := \sup_{x > 0} (U(x) - xy)\,,\quad \; y > 0\;.
$$
The supremum is attained at the point $I(y) := (U')^{-1}(y)$ and a
direct computation shows that  $I(y) = y^{\frac{1}{p-1}}$ and $
\wt U (y) =  - \frac{y^q}{q}$
where $q:=\frac{p}{p-1}<0$. We also have the conjugate relation
\begin{equation}\label{relation:conjugate:UtoUtilde}
    U(x) = \inf_{y>0} (\wt U (y)  + xy) \;, \quad x > 0 \;.
\end{equation}
Before studying the dual problem, we characterize the set of  equivalent
martingale measures which is used to introduce the dual problem.

\subsection{Characterization of the set of equivalent
martingale measures}  \label{MP}
%Our first goal is to characterize
The set $\M(\bP)$ of  equivalent
martingale measures (e.m.m.) is
$$\M(\bP):= \{\bQ ~ | ~ \bQ \sim \bP, RS \text{ is a } (\bQ,\bG)-\text{local martingale  }\}.$$
The dynamics of the discounted price of the risky
asset $\tilde{S} := R S$ is given by %derived via It\^o's formula
\begin{equation} \label{discontS}
d \tilde{S}_t = %\tilde{S}_t \big((\nu_t - r_t) dt + \sigma_t dB_t\big) =
\sigma_t \tilde{S}_t (dB_t + \theta_t dt)\;,
\end{equation}  where $\theta_t := \frac{\nu_t-r_t}{\sigma_t}$ is
the risk premium.

Let $\bQ$ be a probability measure equivalent to $\bP$, defined by its Radon-Nikodym density $$
{d\bQ} \big|_{\G_t} = L_t ^\bQ{d\bP}\big|_{\G_t} \;,$$ where $L^\bQ$ is a positive
$\bG$-martingale with $L_0^\bQ=1$.

According to  the Predictable Representation Theorem (see Kusuoka
\cite{Ku}), and using the fact that $L^\bQ$ is positive, there
exists a pair $(a, \gamma)$ of $\bG$-predictable processes satisfying
$\gamma_t > -1$ for any $t \in [0,T]$ such that
%any square-integrable $\bG$-martingale can be written as the sum of its initial value, a continuous martingale, given by a stochastic integral w.r.t. the Brownian motion, and a discontinuous martingale, given by a stochastic integral w.r.t. the compensated process $M$. Then, the process $L^Q$ can be represented as
%\begin{equation}
%L_t^Q = 1 + \int_0^t \tilde{a}_s dB_s + \int_0^t \tilde{\gamma}_s d M_s,
%\end{equation}
%with some $\bG$-predictable processes $\tilde{a}, \tilde{\gamma}$. Since $L_t^Q$ is positive, by setting  $a_t = \frac{\tilde{a}_t}{L_{t^-}^Q}$ and $\gamma_t = \frac{\tilde{\gamma}_t}{L_{t^-}^Q}$, we obtain the following representation
$$
dL_t^\bQ =   {L_{t^-}^\bQ} ( a_t dB_t + \gamma_t d M_t) \;.
$$
From %\textcolor{red}{>From}
Girsanov's theorem, the process $\widehat B$ defined by
$$
\widehat{B}_t := B_t - \int_0^{t} a_s ds
$$
is a $(\bQ,\bG)$-Brownian motion, and the process $\widehat M$ defined
by
$$
\widehat{M}_t := M_t - \int_0^{t} \gamma_s  \lambda_s^\bG ds = H_t -
\int_0^{t} (1 + \gamma_s)  \lambda_s^\bG ds
$$
is a $(\bQ,\bG)$-discontinuous martingale, orthogonal to $\widehat{B}$. %\textcolor{red}{est-ce utile et vrai?}

Using \eqref{discontS}, we notice that if a
probability measure $\bQ$ is an e.m.m., %risk-neutral probability
then
$a_t=-\theta_t$ for any $t \in [0,T]$.
% and  $\gamma_t $ is any $\bG$-predicable process with $\gamma_t > -1$.

\begin{Lemma} \label{EMMdens}
The set $\M(\bP)$ is determined by all the probability measures $\bQ$ equivalent
to $\bP$, whose Radon-Nikodym density process has the form
\begin{equation*}
L_t^\bQ = \exp \Big( - \int_0^t  \theta_s d B_s   - \frac{1}{2}
\int_0^t \vert\theta_s\vert^2 ds  + \int_0^t \ln(1+
\gamma_s) d H_s   -
 \int_0^t   \gamma_s \lambda_s^\bG ds   \Big)\;,
\end{equation*}
where $\gamma$ is a $\bG$-predictable process satisfying $\gamma_t
> -1$.
\end{Lemma}
To alleviate the notations, for any $\bQ \in \mathcal M(\bP)$, we write
%In what follows, if $\bQ$ is an e.m.m., we set
$L^{\gamma}$ for $L^\bQ$ where $\gamma$ is the process associated to $\bQ$, i.e.,  $$dL_t^\gamma =   {L_{t^-}^\gamma} ( -\theta_t dB_t + \gamma_t d M_t) \;, \quad L^\gamma_0=1.$$

For any $\bQ \in \M(\bP)$, we remark that $RX^{x, \pi, c} + \int_0^. R_sc_s  ds$ is a positive $(\bQ,\bG)\text{-local martingale}$, hence a supermartingale, so we have
\begin{eqnarray*}%\label{SuperMartingale:WealthProcess}
\bE^\bQ \Big(R_TX_T^{x,
\pi,c} + \int_0^T R_sc_s ds\Big) \leq x\;, \quad \forall \; (\pi,c) \in \A(x)\;,
\end{eqnarray*}
where $\bE^{\bQ}$ denotes the expectation w.r.t. the probability measure $\bQ$ or equivalently
\begin{eqnarray}\label{SuperMartingale:WealthProcess}
\bE \Big(R_TL_T^\gamma X_T^{x,
\pi,c} + \int_0^T R_sL_s^\gamma c_s ds\Big) \leq x\;, \quad \forall \; (\pi,c) \in \A(x)\;.
\end{eqnarray}

\subsection{Dual optimization problem}
We now define the dual problem associated to \eqref{pb} according to the standard theory of convex duality. For that, we consider the set $\Gamma$ of dual admissible processes.
\begin{Definition}
The set $\Gamma$ of dual admissible processes is the set of $\bG$-predictable processes $\gamma$
  such that  there exists two constants $A$ and $C$ satisfying $-1 < A\leq\gamma_t\leq  C$ %
for any $t \in [0,\tau]$ and $\gamma_t = 0$ for any $t\in (\tau, T]$.
\end{Definition}

It is interesting to work with this admissible set $\Gamma$ throughout the sequel since, for any $\gamma \in \Gamma$, the process $L^{\gamma}$ is a  positive $\bG$-martingale (indeed, due to the bounds on $\gamma$, the process $L^\gamma$ is a true martingale), and it satisfies the following integrability property which simplifies some proofs in the sequel. Moreover, we consider that $\gamma$ is null after the time $\tau$ since the value of $\gamma$ after $\tau$  does not interfere in the calculus, thus it is possible to fix any value for $\gamma$ after $\tau$. 

\begin{Lemma}\label{uniforme integrable}
For any $\gamma \in \Gamma$, the process $L^\gamma$ satisfies
$$\bE \Big[ \sup_{t \in [0,T]} (L^\gamma_t)^q \Big] < \infty\;.$$
\end{Lemma}

\begin{proof}
From It\^o's formula, we get
$$d(L^\gamma_t)^q = (L^\gamma_{t^-})^q \Big[ \Big( \frac{1}{2}q(q-1) | \theta_t|^2 - q \lambda_t^\bG \gamma_t + \lambda_t^\bG \big( (1+ \gamma_t)^q -1 \big) \Big) dt - q \theta_t dB_t +  \big( (1+ \gamma_t)^q -1 \big)dM_t \Big]\;.$$
This can be written under the following form  \footnote{$\mathcal{E} (Y)$ denotes the Dol\'{e}ans-Dade stochastic exponential process associated to a generic martingale $Y$.}
$$(L^\gamma_t)^q = K_t \mathcal{E}\Big(-\int_0^. q \theta_s dB_s + \int_0^. \big( (1+ \gamma_s)^q -1 \big)dM_s \Big)_t \;,$$
where $K$ is the bounded process defined by
$$K_t := \exp \Big( \int_0^t \Big( \frac{1}{2}q(q-1) | \theta_s|^2 - q \lambda_s^\bG \gamma_s + \lambda_s^\bG \big( (1+ \gamma_s)^q -1 \big) \Big) ds \Big)\;.$$
Therefore,   there exists a positive constant $C$ such that
$$\bE \Big[ \sup_{t \in [0,T]} (L^\gamma_t)^q \Big]\leq  C\bE\Big[ \sup_{t \in [0,T]}  \mathcal{E}\Big(-\int_0^. q \theta_s dB_s \Big)_t \Big]\;.$$
We conclude by using the Burkholder-Davis-Gundy inequality.
\end{proof}

 From the conjugate relation (\ref{relation:conjugate:UtoUtilde}),
%we know that
%\begin{eqnarray*}
%    U(x) = \inf_{y>0} (\wt U (y)  + xy), \; x > 0 \;,
%\end{eqnarray*}
%therefore
we get for any $\eta>0$, $\gamma\in\Gamma$ and $(\pi, c)\in\mathcal{A}(x)$
\begin{eqnarray*}
\begin{split}
	\bE \Big[ \int_0^T U(c_s) ds + U(X^{x, \pi, c}_T) \Big] &\leq \bE \Big[ \int_0^T \wt U(\eta R_s L_s^{\gamma}) ds + \wt U(\eta R_T L_T^{\gamma}) \Big]\\
	&\quad+ \eta\bE\Big[\int_0^T  R_s L_s^{\gamma}c_sds + R_T L_T^{\gamma}X^{x, \pi, c}_T\Big]\;.
\end{split}
\end{eqnarray*}
%Since $\bE\Big[\int_0^T  R_s L_s^{\gamma}c_sds + R_T L_T^{\gamma}X^{x, \pi, c}_T\Big]\leq x$
Using (\ref{SuperMartingale:WealthProcess}), the previous inequality gives for any $\eta>0$, $\gamma\in\Gamma$ and $(\pi, c)\in\mathcal{A}(x)$
\begin{eqnarray*}
\begin{split}
	\bE \Big[ \int_0^T U(c_s) ds + U(X^{x, \pi, c}_T) \Big] &\leq \bE \Big[ \int_0^T \wt U(\eta R_s L_s^{\gamma}) ds + \wt U(\eta R_T L_T^{\gamma}) \Big] + \eta x\;.
\end{split}
\end{eqnarray*}
Therefore, the following inequality holds for any $(\pi, c)\in\mathcal{A}(x)$
\begin{eqnarray*}
	\bE \Big[ \int_0^T U(c_s) ds + U(X^{x, \pi, c}_T) \Big] &\leq& \inf_{\eta>0, \gamma \in \Gamma} \Big(\bE \Big[ \int_0^T \wt U(\eta R_s L_s^{\gamma}) ds + \wt U(\eta R_T L_T^{\gamma}) \Big] + \eta x\Big)\;.  %\\
\end{eqnarray*}	
We thus obtain
\begin{eqnarray}\label{eq:primal:dual:relation:1}
	\sup_{(\pi,c) \in \A(x)}\bE \Big[ \int_0^T U(c_s) ds + U(X^{x, \pi, c}_T) \Big] \leq \inf_{\eta>0, \gamma \in \Gamma} \Big(\bE \Big[ \int_0^T \wt U(\eta R_s L_s^{\gamma}) ds + \wt U(\eta R_T L_T^{\gamma}) \Big] + \eta x\Big)\;.  %\\
\end{eqnarray}
We introduce the dual problem for any $\eta>0$
\begin{eqnarray*}%\label{pbdual}
\tilde{V}(\eta) &=& \inf_{\gamma \in \Gamma} \bE \Big[ \int_0^T \tilde U(\eta R_s L_s^\gamma) ds + \tilde{U}(\eta
R_T L_T^\gamma)  \Big] \\
&=& -\frac{\eta^q}{q}\inf_{\gamma \in \Gamma} \bE \Big[\int_0^T (R_s L^\gamma_s)^q ds+  (
R_T L_T^\gamma)^q  \Big]\;.
\end{eqnarray*}
We thus consider the following optimization problem
\begin{equation*} \label{dual}   \inf_{\gamma \in \Gamma} \bE \Big[
 \int_0^T (R_s L^\gamma_s)^q ds +( R_T L_T^\gamma)^q \Big] \;.
\end{equation*}

To solve this problem we use a similar approach to the one used in  Cheridito and Hu \cite{Cheridito} which is linked to the dynamic programming principle. More precisely, we look for a family of processes $\{(J^{(d)}_t(\gamma))_{t \in [0,T]} ~: ~\gamma \in \Gamma\}$, called the conditional gains, satisfying the following conditions
\begin{enumerate}[(i)]
\item $J^{(d)}_T(\gamma)= (R_T L_T^\gamma)^q + \int_0^T (R_s L^\gamma_s)^q ds$, for any $\gamma\in \Gamma$.

\item $J^{(d)}_0(\gamma_1)=J^{(d)}_0(\gamma_2)$, for any $\gamma_1,\gamma_2\in\Gamma$. 

\item $J^{(d)}(\gamma)$ is a $\bG$-submartingale for any $\gamma\in \Gamma$.

\item  There exists some $\gamma^*\in \Gamma$ such that $J^{(d)}(\gamma^*)$ is a $\mathbb{G}$-martingale.
\end{enumerate}

\noindent Under these conditions, we have 
$$J^{(d)}_0(\gamma^*)   =   \inf_{\gamma \in \Gamma} \bE \Big[ \int_0^T (R_s L^\gamma_s)^q ds+  (
R_T L_T^\gamma)^q \Big]\;.$$
Indeed, using (i) and (iii), we have
\begin{equation}\label{inegHIM3}
J^{(d)}_0(\gamma)  \leq \bE\big[ J^{(d)}_T(\gamma) \big]~=~\bE\Big[  \int_0^T (R_s L^\gamma_s)^q ds + (R_T L_T^\gamma)^q   \Big]\;,
\end{equation}
for any $\gamma \in \Gamma$. Then, using (i) and (iv), we have 
\begin{equation}\label{inegHIM2}
J^{(d)}_0(\gamma^*)  =  \bE\big[ J^{(d)}_T(\gamma^*) \big]~=~\bE\Big[  \int_0^T (R_s L^{\gamma^*}_s)^q ds + (R_T L_T^{\gamma^*})^q  \Big]\;.
\end{equation}
Therefore, from (ii), \reff{inegHIM3} and \reff{inegHIM2}, we get for any $\gamma \in \Gamma$
$$
\bE\Big[ \int_0^T (R_s L^{\gamma^*}_s)^q ds + (R_T L_T^{\gamma^*})^q   \Big] =J^{(d)}_0(\gamma^*)  =  J^{(d)}_0(\gamma) \leq \bE\Big[  \int_0^T (R_s L^\gamma_s)^q ds + (R_T L_T^\gamma)^q   \Big]\;.
$$
We can see that
$$J^{(d)}_0(\gamma^*) = \inf_{\gamma \in \Gamma} \bE \Big[\int_0^T (R_s L^\gamma_s)^q ds + 
( R_T L_T^\gamma)^q \Big]\;.$$

We now construct a family of processes $\{(J^{(d)}_t(\gamma))_{t \in [0,T]} ~: ~\gamma \in \Gamma\}$ satisfying the previous conditions using BSDEs. For that we look for $J^{(d)}(\gamma)$ under the following form, which is based on the dynamic programming principle, 
\begin{eqnarray}\label{eq:J:dual}
J^{(d)}_t(\gamma) =  \int_0^t (R_s L^\gamma_s)^q ds + (R_t L^\gamma_t)^q \Phi_t\;, \quad t \in [0,T]\;,
\end{eqnarray}
where $(\Phi, \widehat \varphi, \tilde \varphi)$ is solution in $\mathcal{S}^{\infty}_{\bG}(0, T) \times
\mathcal{H}^2_{\bG}(0, T) \times \mathcal{H}^2_{\bG}(M)  $ to 
\begin{equation}  \label{BSDEPhi}
\Phi_t= 1 - \int_t^T f (s,\Phi_s, \widehat \varphi_s, \tilde
\varphi_s) ds - \int_t^T \widehat \varphi_s dB_s - \int_t^T \tilde
\varphi_s dH_s \;,
\end{equation}
where $f$ is to be determined such that (iii) and (iv) above hold.  In order to determine $f$, we write $J^{(d)}(\gamma)$ as the sum of a martingale and a non-decreasing process that is null for some $\gamma^* \in \Gamma$.\\

\noindent Applying integration by parts formula leads us to
\begin{equation} \label{expldual}
\begin{split}
d [( R_t L_t^{\gamma})^q ] & =   (R_t
L_{t^-}^{\gamma})^q  \Big[ \Big( \frac{1}{2}q (q-1)
|\theta_t|^2 + \lambda_t \big((1+ \gamma_t)^q -1 \big) -q (\lambda_t^\bG\gamma_t + r_t)  \Big) dt  \\
& \quad  - q \theta_t d B_t + \big((1+ \gamma_t)^q -1 \big) d M_t \Big]\;.
\end{split}
\end{equation}
Taking into account (\ref{expldual}) and applying integration by parts formula for the product of processes $ ( R
L^{\gamma})^q$ and $\Phi$, we get
 $$ d J^{(d)}_t(\gamma)
=
 (R_t L_{t^-} ^{\gamma})^qA_t^\gamma dt+ (R_t L_{t^-} ^{\gamma})^q \Big(  (\widehat \varphi_t-  q \theta_t \Phi_{t}   ) d B_t
+   \big[(\Phi_{t^-}  + \tilde{\varphi}_t) (1+ \gamma_t)^q  -
\Phi_{t^-}\big]  d M_t\Big)\;,
 $$
where the predictable finite variation part of $J^{(d)}(\gamma)$  is given by $\int_0^\cdot   (R_s L_{s^-} ^{\gamma})^qA_s^\gamma ds$, where
\begin{equation*}
\begin{split}
A_t^{\gamma} &  := %(R_t L_{t^-} ^{\gamma})^q \Big[  1+f_t - q r_t
\lambda^\bG_ta_t(\gamma_t) +  1+ f (t,\Phi_t, \widehat \varphi_t, \tilde \varphi_t)  - q r_t
\Phi_{t^-}+\frac{1}{2}q (q-1)  |\theta_t|^2 \Phi_{t^-} -
\lambda_t^\bG \Phi_{t^-}
    - q \theta_t \widehat{\varphi}_t   \;,
\end{split}
\end{equation*}
with 
\begin{equation}\label{function:drift:A}
a_t(x) :=  (\Phi_{t^-} +
\tilde{\varphi}_t) (1+ x)^q - q\Phi_{t^-} x \;,  \quad t \in [0,T]\;.
\end{equation}
In order to obtain a non-negative process $A^\gamma$ for any $\gamma \in \Gamma$ (to satisfy the condition (iii)) and that is null for some $\gamma^*\in \Gamma$ (to satisfy the condition (iv)), it is obvious that the family $\{(A_t^{\gamma})_{t \in [0,T]} ~: ~\gamma \in \Gamma\}$ has to satisfy  $\min_{\gamma \in\Gamma} A_t^{\gamma} = 0$. Assuming that there exists a positive constant $C$ such that $\Phi_{t}\geq C$ and $\Phi_{t^-} + \tilde{\varphi}_t\geq C$ for any $t\in[0, \tau)$,
%the processes $(\Phi_{t^-})_{t\geq 0}$ and  $(\Phi_{t^-} + \tilde{\varphi}_t)_{t\geq 0}$ take positive values on the set $(t < \tau)$,
we remark that the minimum is attained for $\gamma^*$ defined by
$$\gamma^*_t := \Big(\frac{\Phi_{t^-}}{\Phi_{t^-} + \tilde{\varphi}_t}\Big)^{\frac{1}{q-1}} - 1\; $$ so that
$$\underline a_t := \min_{x>-1} a_t(x) = (1-q)\Phi_{t^-}^{p} (\Phi_{t^-} + \tilde \varphi_t)^{1-p} + q \Phi_{t^-}\;.$$
This leads to the following choice for the generator $f$
  \begin{eqnarray}\label{genedual}
 f (t,y, z, u) & = & \Big(q r_t - \frac{1}{2}q (q-1)  |\theta_t|^2  + (1-q)\lambda_t^\bG\Big) y+ q \theta_t z  \nonumber\\
& &\quad  - (1-q)\lambda_t^\bG   (y + u)^{1-p} \; y^p - 1 \;.% \nhn{-\lambda_t^\bG qx}.
\end{eqnarray}
\subsection{Solution of the BSDE (\ref{BSDEPhi})}
We remark that the obtained generator (\ref{genedual}) is non standard since it involves in particular the term  $(y+u)^{1-p}y^p$.
We shall prove the following result
\begin{Theorem}\label{theoreme existence}
The BSDE
\begin{equation} \label{BSDEPhi0}
\begin{split}
\Phi_t & = 1 - \int_t^T \Big( \big(q r_s - \frac{1}{2}q (q-1)  |\theta_s|^2  +
(1-q) \lambda_s^\bG \big) \Phi_{s} + q \theta_s \widehat{\varphi}_s   \\
& \quad  - (1-q) \lambda_s^\bG  (\Phi_{s} +
\tilde{\varphi}_s)^{1-p} \;  \Phi_{s}^p -1 \Big) ds - \int_t^T
\widehat{\varphi}_s d B_s - \int_t^T \tilde{\varphi}_s d H_s\;,
\end{split}
\end{equation}
admits a solution $(\Phi,  \widehat{\varphi} ,  \tilde{\varphi}) $
belonging to $\mathcal{S}^{\infty,+}_{\bG}(0, T) \times
\mathcal{H}^2_{\bG}(0,T)   \times
 \mathcal{H}^2_{\bG}(M)$,
 such that $\Phi_{t^-} + \tilde{\varphi}_t \geq 1$.
\end{Theorem}

%\subsection{Solution of the BSDE \reff{BSDEPhi0}}
%We now prove Theorem \ref{theoreme existence} via the decomposition procedure introduced in \cite{KL}.
We use the decomposition procedure introduced in \cite{KL} to prove Theorem \ref{theoreme existence}. For that, we transform
the BSDE \eqref{BSDEPhi0} into a recursive system of  Brownian
BSDEs. %. For that we introduce the two following Brownian BSDEs
In a first step, for each $u \in [0,T]$, we prove that the following BSDE has a solution on the time interval $[u,T]$ 
\begin{equation} \label{BSDEafter}
\left\{\begin{array}{rcl}
d \Phi^1_t (u)  &= & \Big[ \big( q r^1_t(u) - \frac{1}{2}q (q-1)  |\theta^1_t(u)|^2 \big) \Phi^1_{t}(u) + q \theta^1_t(u) \widehat{\varphi}^1_t(u) -1 \Big]dt  + \widehat{\varphi}^1_t(u) d B_t\;,\\
&&\\
\Phi^1_T(u) & =&1\;,
\end{array}\right.
\end{equation}
and that the initial value   $\Phi^1_u(u)$  of this BSDE is   $\mathcal{F}_u$-measurable. Then, in a second step, we prove that the following BSDE has a solution on the time interval $[0,T]$ 
\begin{equation} \label{BSDEbefore}
\left\{\begin{array}{rcl}
 d \Phi^0_t
%& =  \Big[ \big( q r^0_t - \frac{1}{2} q (q-1) \theta_t^2  + \lambda_t \big) \Phi^0_{t}+ q \theta_t \widehat{\varphi}^0_t  \\
%& \quad  - (1-q) \lambda_t \left( \Phi^0_t + (\Phi^1_t (t) - \Phi^0_t)
%\right)^{1-p}  (\Phi^0_t \vee 0)^p  \Big]dt + \widehat{\varphi}^0_t d B_t  \\
& = &\Big[ \big( q r^0_t - \frac{1}{2} q (q-1) |\theta^0_t|^2 +(1-q)\lambda_t^\ff \big) \Phi^0_{t}+ q \theta^0_t \widehat{\varphi}^0_t  \\
 && \quad - (1-q)  \lambda_t^\ff( \Phi^1_t (t)  )^{1-p}  (\Phi^0_t )^p  \big) -1 \Big]dt + \widehat{\varphi}^0_t d B_t\;, \\
 \Phi^0_T &= &1\;,
\end{array}\right.
\end{equation}
where $\Phi^1$ is part of the solution of the BSDE
\eqref{BSDEafter}. 

\begin{Proposition}\label{existence BSDEafter}
For any $u \in [0,T]$, the BSDE (\ref{BSDEafter}) admits a unique solution $(\Phi^1(u), \widehat
\varphi^1(u))\in \mathcal{S}^{\infty}_{\ff}(u, T) \times
\mathcal{H}^2_{\ff}(u,T)$. Furthermore,   $1
\leq \Phi^1_t(u) \leq C$ for any $t \in [u,T]$
where $C$ is a constant which does not depend on $u$.
\end{Proposition}

\begin{proof} Let us fix $u\in[0, T]$. Since the BSDE (\ref{BSDEafter}) is linear with bounded coefficients,
the solution  $(\Phi^1 (u), \widehat{\varphi}^1(u)) \in
\mathcal{S}^{\infty}_{\ff}(u, T) \times \mathcal{H}^2_{\ff}(u,T)$ is given by
\begin{equation}   \label{eqPhi1}
\Phi^1_t (u) =  \bE \Big[\Gamma_T^t(u) + \int_t^T \Gamma_s^t(u) ds \Big| \F_t   \Big]\;,\quad t \in [u,T]\;,
\end{equation}
where for a fixed $t\in[u, T]$, $(\Gamma_s^t(u))_{t \leq s \leq T}$ stands for the adjoint process defined by 
$$\Gamma_s^t (u)  = \exp \Big( \int_t^s \big(-q r^1_v(u) + \frac{1}{2} q (q-1) |\theta^1_v(u)|^2 \big) dv \Big) \E \Big(-\int_t^{\cdot} q \theta^1_v(u) d B_v \Big)_s. \\
%& = \exp \left( \int_t^s \big(q r^1_v(u) + \frac{q}{2} \theta_v^2 \big) dv - \int_t^{s} q \theta_v d B_v \right).
$$
To prove that $\Phi^1$ is uniformly bounded, we introduce the probability  measure  $\bP^{  u}$,  defined on $\F_t$, for $t\leq T$,
by its Radon-Nikodym density $Z_t(u) := \E
(-\int_0^{\cdot} q \theta^1_v(u) d B_v)_t$, which is a true martingale,  and we denote
by $\bE^{  u}$ the expectation under this probability. Then, by
virtue of the formula \eqref{eqPhi1} and Bayes' rule, we get
\begin{eqnarray*}
\Phi^1_t (u) & = &\bE^{  u} \Big[ \exp \Big( \int_t^T \big(-q r^1_s(u) +
\frac{1}{2} q (q-1) |\theta^1_s(u)|^2 \big) ds \Big) \Big| \F_t
\Big]\\
&& +\; \bE^{  u} \Big[ \int_t^T \exp \Big( \int_t^s \big(-q r^1_v(u) + \frac{1}{2} q (q-1) |\theta^1_v(u)|^2 \big) dv \Big) ds \Big| \F_t  \Big]\;. \\
\end{eqnarray*}
From \textbf{(H3)} and \textbf{(H4)}, and since $q<0$, there exists a positive constant $C$ which is independent of $u$ such that $1 \leq \Phi^1_t(u) \leq C$ for any $t\in[u,T]$.
\end{proof}
%\mb{explain that $\Phi ^1_t(t)$ can be defined for any $t$ as}
\begin{Proposition}\label{existence BSDEbefore}
The BSDE (\ref{BSDEbefore}) admits a unique solution $(\Phi^0, \widehat
\varphi^0)\in \mathcal{S}^{\infty,+}_{\ff}(0, T) \times
\mathcal{H}^2_{\ff}(0, T)$. \end{Proposition}

\begin{proof}
The generator of the BSDE (\ref{BSDEbefore}) is not defined on the whole space $[0, T]\times\Omega\times\R\times\R$ and the generator is not classical. So the proof of this proposition will be performed in several steps.  We first introduce a modified BSDE where the term $y^p$ is replaced by $(y \vee m)^p$ (where $m$ is a positive constant which is defined later) to ensure that the generator is well defined on the whole space $[0, T]\times\Omega\times\R\times\R$. We then prove via a comparison theorem that the solution of the modified BSDE satisfies the initial BSDE. In the last step, we prove the uniqueness of the solution.\\
{\bf Step 1.}  \textit{Introduction of the modified BSDE}.

We consider
\begin{equation} \label{BSDEmodified}
\left\{\begin{array}{rcl}
 -d Y_t
& = &\bar{g} (t, Y_t, \widehat{y}_t)dt - \widehat{y}_t d B_t\;,\\
Y_T&=&1\;,
\end{array}\right.
\end{equation}
where the generator $\bar g$ is given by
$$
\bar{g} (t,y, z) := 1+\Big(  \frac{1}{2} q (q-1)
|\theta^0_t|^2 -q r^0_t - (1-q)\lambda_t ^\ff\Big) y - q \theta^0_t z
   + (1-q)  \lambda_t^\ff ( \Phi^1_t (t)  )^{1-p} (y \vee m)^p\;,
$$
with $m:= \exp\big( (q-1) \Lambda \big)$, and $\Lambda$ is a constant such that $\lambda^\ff_t \leq \Lambda$ for any $t \in [0,T]$.

Since $p\in(0, 1)$, there exists a positive constant $C$ such that  $(y \vee m)^p \leq C(1 + |y|)$. % (recall $ p \in (0,1)$).
We also have $\Phi^1 (.)$ is uniformly bounded, and using  assumptions
${\bf (H2)}$, ${\bf (H3)}$ and ${\bf (H4)}$ we obtain that
$\bar{g}$ has linear growth uniformly w.r.t. $y$.
%Moreover, the mapping $(y, z) \to \bar{g} (t,y, z) $ is
%continuous, $dt \otimes dP$ $a.s.$
It follows from
  Fan and Jiang \cite{FanJiang} that the BSDE
\eqref{BSDEmodified} has a unique solution $(Y, \widehat y) \in
\mathcal{S}^{\infty}_{\ff}(0, T) \times \mathcal{H}^2_{\ff}(0,T)$.\\

For the convenience of the reader, we recall the Fan and Jiang  conditions, which, in our setting, are obviously satisfied.
The solution of the  BSDE  
$$-dY_t= f(t,Y_t,\widehat y_t)dt -\widehat y_t dB_t\;,\quad Y_T=1$$
 is unique if:\\

(1) the process $(f(t, 0, 0))_{t\in [0,T ]} \in L^2(0,T)$,

(2) $(d\bP\times dt)$  a.s., $(y, z)\rightarrow f (\omega, t, y, z)$ is continuous,

(3) $f$ is monotonic in $y$, i.e., there exists a constant $\mu \geq 0$, such that, $(d\bP\times dt)$  a.s.,
$$\forall y_1,y_2,z,  \Big(f(\omega, t, y_1, z) - f(\omega, t, y_2, z)  \Big)(y_1-y_2) \leq \mu  (y_1-y_2)^2\;,$$

(4) $f$ has a general growth with respect to $y$, i.e., $(dP\times dt)$ .a.s.,
$$\forall y, \vert f(\omega,t,y,0)\vert \leq \vert f(\omega,t,0,0)\vert +\varphi(|y|)$$
where $\varphi : \R \rightarrow \R^+$ is an increasing continuous function,

(5) $f$ is uniformly continuous in $z$ and uniform w.r.t. $(\omega, t, y)$, i.e., there exists a continuous, non-decreasing
function $\phi$ from $\R^+$ to itself with at most linear growth and $ƒ \phi(0) = 0$ such that $(d\bP\times dt)$  a.s.,
$$\forall  y,z_1,z_2, \quad \vert f(\omega,t,y,z_1)-f(\omega,t,y,z_2)\vert \leq \phi(\vert z_1-z_2\vert )\;. $$

{\bf Step 2.} \textit{Comparison}.

We now show that the solution of the BSDE \eqref{BSDEmodified} is lower bounded by $m$, and this is accomplished via a comparison result for solutions of Brownian BSDEs. We remark that the following inequality holds
$$
\bar{g} (t,y, z) \geq \Big(  \frac{1}{2} q (q-1)
|\theta^0_t|^2 -q r^0_t - (1-q)\lambda_t^\ff \Big) y - q \theta^0_t z =: g(t, y, z)\;.%  := {g} (t,y, z)\;.
$$
Therefore, we introduce the following linear BSDE
\begin{equation} \label{BSDEunder} \left\{\begin{array}{rcl}
-d Z_t & = & {g} (t,Z_t, \widehat {z}_t)dt -  \widehat {z}_t d B_t\;, \\
 Z_T &=& 1\;.
 \end{array}\right.
\end{equation}
In the same way as we proceed with the BSDE \eqref{BSDEafter}, we have an explicit form of the solution of the BSDE (\ref{BSDEunder}) given by
$$
Z_t  = \bE \big( {\Upsilon}_T^t \big| \F_t   \big)\;,
$$
where  $(\Upsilon_s^t)_{t \leq s \leq T}$ stands for the solution
of the linear SDE
$$
d\Upsilon_s^t = \Upsilon_s^t   \Big[ \Big(   \frac{1}{2}q (q-1)
 |\theta^0_s|^2 - q r^0_s -(1-q)\lambda_s^\ff \Big)ds - q \theta^0_s
d B_s \Big] \;, \quad \Upsilon_t^t =1 \;.
$$
We can rewrite the solution of the BSDE (\ref{BSDEunder}) under the following form
$$
Z_t = \bE^{*} \Big[ \exp \Big( \int_t^T \big( 
\frac{1}{2}q (q-1)  |\theta^0_s|^2 -q r^0_s - (1-q)\lambda_s^\ff \big) ds
\Big) \Big| \F_t  \Big] \;,
$$
where $\bE^{*}$ is the expectation under the probability $\bP^{*}$  %which is
defined by its Radon-Nikodym density $d \bP^{*}|_{\F_t} =  \E (-\int_0^{\cdot} q \theta^0_v d B_v )_t d\bP|_{\F_t} $ for any $t \in [0,T]$. %    $\frac{d \bP^{0,*}}{d \bP} | _{\F_T} := \E \left(-\int_0^{\cdot} q \theta^0_v d B_v \right)_T$.
 By virtue of the assumption {\bf (H4)}, it follows that 
$$
Z_t \geq \bE^{*} \Big[ \exp \Big( -\int_t^T (1-q)\lambda_s^\ff
ds\Big) \Big| \F_t  \Big] \geq m \;.
$$
From the comparison theorem for Brownian BSDEs, we obtain
%We now apply a  comparison result for solutions of BSDEs (see a slightly modified version of Proposition 1.2 in El Karoui\emph{ et al.} %, Hamadene, Matoussi
%\cite{EHM},) from which we deduce
$$
Y_t \geq Z_t \geq m \;,
$$
which implies that $Y_t \vee m= Y_t$ for any $t\in [0,T]$. Therefore, $(Y, \widehat y)$ is a solution of the BSDE (\ref{BSDEbefore}) in $\mathcal{S}^{\infty,+}_{\ff}(0, T) \times
\mathcal{H}^2_{\ff}(0, T)$.\\

{\bf Step 3.}  \textit{Uniqueness of the solution}.
Suppose that the BSDE \eqref{BSDEbefore} has two solutions  $(Y^1,Z^1)$ and $(Y^2,Z^2)$ in $\mathcal{S}^{\infty,+}_{\ff}(0, T) \times \mathcal{H}^2_{\ff}(0,T)$. Thus, there exists a positive constant $c$ such that $Y^1_t \geq c$ and $Y^2_t \geq c$ for any $t \in [0,T]$. In this case, $(Y^1,Z^1)$ and $(Y^2,Z^2)$ are solutions of the following BSDE
\begin{equation*} 
\left\{\begin{array}{rcl}
 -d Y_t
& = &h (t, Y_t, \widehat{y}_t)dt - \widehat{y}_t d B_t\;,\\
Y_T&=&1\;,
\end{array}\right.
\end{equation*}
where the generator $h$ is given by
$$
h (t,y, z) := 1+\Big(-q r^0_t + \frac{1}{2} q (q-1)
|\theta^0_t|^2 - (1-q)\lambda_t ^\ff\Big) y - q \theta^0_t z
   + (1-q)  \lambda_t^\ff ( \Phi^1_t (t) )^{1-p} (y \vee c)^p\;.
$$
 From \cite{FanJiang}, we know that this   BSDE
admits a unique solution, therefore we get $Y^1=Y^2$.
\end{proof}
\vspace{2mm}

\noindent We  are now able to prove Theorem \ref{theoreme existence}.

\begin{proof} From Propositions \ref{existence BSDEafter} and \ref{existence BSDEbefore} and Theorem 3.1 in \cite{KL}, we obtain that the BSDE \eqref{BSDEPhi0} admits a solution $(\Phi,  \widehat{\varphi} ,  \tilde{\varphi}) $
belonging to $\mathcal{S}^{\infty}_{\bG}(0, T) \times
\mathcal{H}^2_{\bG}(0,T)   \times
 \mathcal{H}^2_{\bG}(M)$ given by
\begin{equation} \label{solPhi} \begin{split}
\Phi_t  &= \Phi^0_t \1_{t <\tau } + \Phi^1_t(\tau)  \1_{t \geq \tau}\;,\\
\widehat \varphi_t &= \widehat \varphi^0_t \1_{t \leq \tau } + \widehat \varphi^1_t(\tau)  \1_{t > \tau}\;,\\
\tilde{\varphi}_t &= (\Phi^1_t(t) - \Phi^0_t)  \1_{t \leq \tau}\;.
\end{split}
\end{equation}
Note that $\widehat \varphi$ and $\tilde{\varphi}$ are $\bG$-predictable
processes. Moreover, from Propositions \ref{existence BSDEafter} and
\ref{existence BSDEbefore}, there exists a positive constant $C$
such that $\Phi_t \geq C$.
% and that $\Phi \in
%\mathcal{S}^{\infty}_{\bG}(0, T)$. By definition of $\widehat \varphi$
%it is obvious that $\widehat \varphi \in \mathcal{H}^2_{\bG}(0,T)$ and
%$\tilde \varphi \in \mathcal{H}^2_{\bG}(M)$ from Propositions
%\ref{existence BSDEafter} and \ref{existence BSDEbefore}.
%Moreover,
We also remark that
$$
\Phi_{t^-} + \tilde{\varphi}_t =     \Phi^1_t(t)  \1_{t \leq \tau}    + \Phi^1_t(\tau)  \1_{t > \tau} = \Phi^1_t(t \wedge \tau)\;,
$$
which implies
that $\Phi_{t^-} +\tilde \varphi_t \geq 1$.
\end{proof}
\begin{Remark}\label{Remark:Deterministic:Case}
We remark that if $r^1$ and $\theta^1$ are deterministic, $\Phi^1$ is deterministic. Moreover, if $r^0$, $\theta^0$ and $\lambda^\ff$ are deterministic, $\Phi^0$ is deterministic. More precisely, $\widehat \varphi^1_t(u) = \widehat \varphi^0_t =0$, and the BSDEs (\ref{BSDEafter}) and (\ref{BSDEbefore}) turn into ODEs
\begin{equation*}
\left\{\begin{array}{rcl}
d \Phi^1_t (u)  &= & \Big[ \big( q r^1_t(u) - \frac{1}{2}q (q-1)  |\theta^1_t(u)|^2 \big) \Phi^1_{t}(u) ) -1 \Big]dt  \;,\\
&&\\
\Phi^1_T(u) & =&1 \;,
\end{array}\right.
\end{equation*}
and
\begin{equation*}
\left\{\begin{array}{rcl}
 d \Phi^0_t
& = &\Big[ \Big( q r^0_t - \frac{1}{2} q (q-1) |\theta^0_t|^2 +(1-q)\lambda_t^\ff \Big) \Phi^0_{t}   - (1-q)  \lambda_t^\ff ( \Phi^1_t (t)  )^{1-p}  (\Phi^0_t )^p  \big) -1 \Big]dt \;,\\
 \Phi^0_T &= &1\;,
\end{array}\right.
\end{equation*}
with an explicit solution for the first equation.
\end{Remark}

\begin{Remark}
If $r_t^1(u) = r^0_t$ for any $0\leq u\leq t\leq T$, there is no change of regime. Our result is coherent with that obvious observation, since, in that case,  we have that $\theta_t^1(u) = \theta^0_t$
for any $0\leq u\leq t\leq T$ which implies $\Phi^0_t = \Phi^1_t(t)$ for any $t\in[0, T]$.
\end{Remark}
\subsection{A verification Theorem}
We now turn to the sufficient condition of optimality. In this part, we prove that the family of processes $\{(J^{(d)}_t(\gamma))_{t \in [0,T]} ~: ~\gamma \in \Gamma\}$ defined by $J^{(d)}(\gamma) :=\int_0^.( R_s L^{\gamma}_s)^q ds +  ( R L^{\gamma})^q \Phi$ with $\Phi$ defined by (\ref{solPhi}) satisfies the conditions (i), (ii), (iii) and (iv). By construction, $J^{(d)}(\gamma)$ satisfies the conditions (i) and (ii). As explained previously a candidate to be an optimal $\gamma $ is a process $\gamma^*$ such that $J^{(d)}(\gamma^*)$ is a $\bG$-martingale, hence this one is
%$$
%\gamma_t^* := \left(\frac{\Phi_{t^-}}{\Phi_{t^-} + \tilde{\varphi}_t}\right)^{\frac{1}{q-1}} - 1.
%$$
\begin{equation} \label{optimalgamma}
\gamma_t^* := \Big(\frac{\Phi_{t^-}}{\Phi_{t^-} + \tilde{\varphi}_t}\Big)^{\frac{1}{q-1}} - 1 \;.
\end{equation}

\begin{Lemma}
The process $\gamma^*$ defined by (\ref{optimalgamma}) is admissible. %\textcolor{red}{est-ce que on doit pr\'eciser admissible $\gamma$?}
\end{Lemma}
\begin{proof} By construction, $\gamma^*$ is $\bG$-predictable.
Moreover, from Theorem \ref{theoreme existence}, we remark that
there exists two constants $A$ and $C$ such that $-1 < A \leq \gamma^*_t \leq C$ for
any $t\in [0,T]$  which implies that $\gamma^*\in\Gamma$.%$\bE[L^{\gamma^*}_T]=1$.
\end{proof}

%Before proving that the process $J^{(d)}(\gamma)$ is a $\bG$-submartingale for any admissible process $\gamma \in \Gamma$ and is a $\bG$-martingale for  $\gamma^*$, w
From the above results, $  J^{(d)}(\gamma) $ is a semi-martingale with a local martingale part and a non-decreasing predictable variation part and $  J^{(d)}(\gamma ^*)$ is a local martingale. 
%We now give an integrability property about the \mb{Radon Nikodym densities  $L^\gamma$}.
%\begin{Lemma}\label{uniforme integrable}
%For any $\gamma \in \Gamma$, the process $L^\gamma$ satisfies
%$$\bE \Big[ \sup_{t \in [0,T]} (L^\gamma_t)^q \Big] < \infty\;.$$
%\end{Lemma}
%
%\begin{proof}
%From It\^o's formula, we get
%$$d(L^\gamma_t)^q = (L^\gamma_{t^-})^q \Big[ \Big( \frac{1}{2}q(q-1) | \theta_t|^2 - q \lambda_t^\bG \gamma_t + \lambda_t^\bG \big( (1+ \gamma_t)^q -1 \big) \Big) dt - q \theta_t dB_t +  \big( (1+ \gamma_t)^q -1 \big)dM_t \Big]\;.$$
%This can be written under the following form
%$$(L^\gamma_t)^q = K_t \mathcal{E}\Big(-\int_0^. q \theta_s dB_s + \int_0^. \big( (1+ \gamma_s)^q -1 \big)dM_s \Big)_t \;,$$
%with $K$ defined by
%$$K_t := \exp \Big( \int_0^t \Big( \frac{1}{2}q(q-1) | \theta_s|^2 - q \lambda_s^\bG \gamma_s + \lambda_s^\bG \big( (1+ \gamma_s)^q -1 \big) \Big) ds \Big)\;.$$
%Therefore,   there exists a positive constant $C$ such that
%$$\bE \Big[ \sup_{t \in [0,T]} (L^\gamma_t)^q \Big]\leq  C\bE\Big[ \sup_{t \in [0,T]}  \mathcal{E}\Big(-\int_0^. q \theta_s dB_s \Big)_t \Big]\;,$$
%since by assumption $\theta$, $\lambda$, $\gamma$ are bounded and also $(1 + \gamma)^q$. We conclude by using BDG inequality.
%\end{proof}
%

\begin{Proposition}\label{prop:verification:martingale}
The process $J^{(d)}(\gamma)$ is a  $\bG$-submartingale for any admissible process $\gamma \in \Gamma$ and is a $\bG$-martingale for  $\gamma^*$ given by \eqref{optimalgamma}.
\end{Proposition}
\begin{proof}
From \eqref{expldual} and \eqref{BSDEPhi0}, we can rewrite the dynamics of $J^{(d)}(\gamma)$ under the following form
$$
d J^{(d)}_t(\gamma)=(R_tL^{\gamma}_{t^-})^q  \big( dM^\gamma_t +  A^\gamma_t dt\big) \;,
$$
where
$$dM^\gamma_t= \big ({\widehat{\varphi}_t} -   q \theta_t   \Phi_{t^-}  \big)d B_t  +   \Big( (1+ \gamma_t)^q\big( \Phi_{t^-}  +  {\tilde{\varphi}_t}\big)   -  \Phi_{t^-}  \Big)  d M_t \;,
$$
and
$$
A_t^{\gamma}   =  { \lambda_t^\bG}  \big[  a_t(\gamma_t)-a_t(\gamma_t^*) \big]\;,
$$
with $a(.)$ defined by (\ref{function:drift:A}). \\
From (\ref{eq:J:dual}), Lemma \ref{uniforme integrable} and since $\Phi\in\mathcal{S}^{\infty}_{\bG}(0, T)$, we remark that for any $\gamma\in\Gamma$
\begin{eqnarray}\label{J:bounded}
\bE \Big[ \sup_{t \in [0,T]} J^{(d)}_t(\gamma) \Big] < \infty\;.
\end{eqnarray}
For any $\gamma\in\Gamma$, we have that $\int_0^.(R_sL^{\gamma}_{s^-})^q dM_s^\gamma$ is a $\bG$-local martingale. Hence, there exists an increasing sequence of $\bG$-stopping times $(T_n)_{n\in\mathbb{N}}$ valued in $[0,T]$ satisfying $\lim_{n\to\infty}T_n = T\;, \;\mathbb{P}-a.s.$ such that
%$(M^{\gamma}_{.\wedge T_n})$ is a $\bG$-martingale. From Lemma \ref{uniforme integrable} and since $\Phi\in\mathcal{S}^{\infty}_{\bG}(0, T)$, we have that
$\int_0^{.\wedge T_n}(R_sL^{\gamma}_{s^-})^q  dM^\gamma_s$ is a $\bG$-martingale for any $n \in \mathbb N$. Therefore, we obtain for any $t\in[0, T]$
\begin{eqnarray*}
	\bE \Big[J^{(d)}_{t\wedge T_n}(\gamma)\Big] = J^{(d)}_0(\gamma) + \bE\Big[\int_0^{t\wedge T_n}(R_sL^{\gamma}_{s^-})^q  A^\gamma_s ds\Big]\;.
\end{eqnarray*}
Since $(RL^{\gamma})^q A^\gamma \geq 0$, from (\ref{J:bounded}) and using the monotone convergence theorem, we obtain
$$\bE \Big[ \int_0^T (R_t L^\gamma_{t^-})^q   A^\gamma_tdt \Big] < \infty \;.$$
From (\ref{J:bounded})  and the previous inequality, we have
$$\bE \Big[ \sup_{t\in[0, T] }   \Big|\int_0^t (R_s L^\gamma_{s^-})^q   dM^\gamma_s \Big| \Big] < \infty\;.$$
It follows that  the local martingale $\int_0^. (R_s L^\gamma_{s^-})^q   dM^\gamma_s$ is a true martingale and the process $J^{(d)}(\gamma)$ is a $\bG$-submartingale for any $\gamma \in \Gamma$. We obtain with the same arguments that the process $J^{(d)}(\gamma^*)$ is a martingale.

\end{proof}

\subsection{Uniqueness of the solution of the BSDE  (\ref{BSDEPhi})}\label{sous section unite}
To solve the dual problem it is not necessary to prove the uniqueness of the solution of the BSDE  (\ref{BSDEPhi}) but this one is useful to characterize the value function of the primal problem in the last part of this paper. To prove the uniqueness we do not use a comparison theorem for BSDE but the following dynamic programming principle.
\begin{Lemma}\label{unicite lemme 1}
Let $Y$ be a process  with $Y_T=1$ such that $\int_0^. (R_s L^\gamma_s)^q ds +(RL^\gamma)Y$ is a $\bG$-submartingale for any $\gamma \in \Gamma$ and there exists $\gamma^\prime \in \Gamma$ such that $\int_0^. (R_s L^{\gamma^\prime}_s )^q ds +(RL^{\gamma^\prime})Y$ is a $\bG$-martingale. Then, we have
$$Y_t = \essinf_{\gamma \in \Gamma} \Big\{ \frac{1}{(R_t L^\gamma_t)^q} \bE\Big[ \int_t^T (R_s L^\gamma_s)^q ds + (R_T L^\gamma_T)^q \Big| \mathcal G_t\Big] \Big\}\;.$$
\end{Lemma}

\begin{proof}
 The following inequality holds for any $\gamma \in \Gamma$
$$Y_t \leq  \frac{1}{(R_t L^\gamma_t)^q} \bE\Big[ \int_t^T (R_s L^\gamma_s)^q ds + (R_T L^\gamma_T)^q \Big| \mathcal G_t\Big] \;.$$
Moreover, we know that
$$Y_t =  \frac{1}{(R_t L^{\gamma^\prime}_t)^q} \bE\Big[ \int_t^T (R_s L^{\gamma^\prime}_s)^q ds + (R_T L^{\gamma^\prime}_T)^q \Big| \mathcal G_t\Big] \;.$$
Therefore, we get
$$Y_t = \essinf_{\gamma \in \Gamma} \Big\{ \frac{1}{(R_t L^\gamma_t)^q} \bE\Big[ \int_t^T (R_s L^\gamma_s)^q ds + (R_T L^\gamma_T)^q \Big| \mathcal G_t\Big] \Big\}\;.$$
\end{proof}
We now prove that any solution of the BSDE  (\ref{BSDEPhi}) satisfies the properties of Lemma \ref{unicite lemme 1}.
\begin{Lemma}\label{unicite lemme 2}
Let $(\Phi, \widehat \varphi, \tilde \varphi) \in \mathcal{S}^{\infty,+}_{\bG}(0, T) \times
\mathcal{H}^2_{\bG}(0,T)   \times
 \mathcal{H}^2_{\bG}(M)$ be a solution of the BSDE  (\ref{BSDEPhi}). Then, the process $\int_0^. (R_s L^\gamma_s)^q ds +(RL^\gamma)\Phi$ is a $\bG$-submartingale for any $\gamma \in \Gamma$ and there exists $\gamma^\prime \in \Gamma$ such that $\int_0^. (R_s L^{\gamma^\prime}_s )^q ds +(RL^{\gamma^\prime})\Phi$ is a $\bG$-martingale.
\end{Lemma}

\begin{proof}
To simplify the notation we denote $\cW^\gamma := \int_0^. (R_s L^\gamma_s)^q ds +(RL^\gamma)\Phi$. From It\^o's formula, we get for any $\gamma \in \Gamma$
$$
\cW^\gamma_t = (R_t L^\gamma_{t^-})^q \Big\{ \lambda_t \big((a_t(\gamma_t) - a_t(\gamma^\prime_t))\big) dt + \big(\widehat \varphi_t - q \theta_t  \Phi_t\big)dB_t +\big( (1+\gamma_t)^q(\Phi_t + \tilde \phi_t)  - \Phi_t\big)dM_t \Big\} \;,
$$
where $a(.)$ is defined by \reff{function:drift:A} and $\gamma^\prime_t := \Big( \frac{\Phi_{t^-}}{\Phi_{t^-} + \tilde \phi_t}\Big)^{p-1}-1$.\\
We know that $\bE[\sup_{0 \leq t \leq T} \cW^\gamma_t] < \infty$ from Lemma \ref{uniforme integrable} and $a_t(\gamma_t) \geq a_t(\gamma^\prime_t)$ for any $\gamma \in \Gamma$ by definition of $\gamma^\prime$. Therefore, using the same arguments as for the proof of Proposition \ref{prop:verification:martingale} we can prove that, for any $\gamma \in \Gamma$, the process $\int_0^. (R_s L^\gamma_s)^q ds +(RL^\gamma)\Phi$ is a $\bG$-submartingale  and  $\int_0^. (R_s L^{\gamma^\prime}_s )^q ds +(RL^{\gamma^\prime})\Phi$ is a $\bG$-martingale.
\end{proof}

We can conclude from Lemmas \ref{unicite lemme 1} and \ref{unicite lemme 2} that there exists a unique solution of the BSDE  (\ref{BSDEPhi}) in $\mathcal{S}^{\infty,+}_{\bG}(0, T) \times
\mathcal{H}^2_{\bG}(0,T)   \times
 \mathcal{H}^2_{\bG}(M)$.

\section{Primal problem and optimal strategy}\label{retour primal}
%\mb{je souhaiterais \'eviter de r\'esoudre \`a la main le pb
%primal et  je voudrais  obtenir les r\'esultats par dualit\'e. Il
%semble que ce soit tr\`es simple L ancienne version peut figurer
%dans la th\`ese de Nam}

In this section, we deduce the solution of the primal problem (\ref{pb})
using the duality result of the previous section, and we characterize the value function associated to the primal problem by the solution of a BSDE which is in relationship with
%the BSDE associated to the
%primal problem (see Hu\emph{ et al.} %, Imkeller, Muller
%\cite{HIM}) satisfies a duality relationship with the %\textcolor{red}{the
%the}
the BSDE (\ref{BSDEPhi0}) associated to the dual problem.\\

The following proposition shows the existence of an optimal solution for the primal problem and characterizes this solution in terms of the solution of the dual problem.
\begin{Proposition}\label{Prop:Existence:OptimalSolution}
The optimal strategy is given by
\begin{equation}  \label{optwealth}
c_t^* =  \big(  \eta^*R_t L_{t^-}^{\gamma^*} \big)^{\frac{1}{p-1}}\;, \quad \pi_t^* = \frac{1}{\sigma_t}\Big(\frac{\wh\varphi_t}{\Phi_{t^-}}+ \frac{\theta_t}{1-p}\Big)\;,\quad t\in[0, T]\;,
\end{equation}
where  $\eta^*$ is defined by
%\textcolor{red}{
\begin{equation} \label{optlambda}
 	\eta^* := \Big(\frac{x} {\bE\Big[\int_0^T(R_t L_t^{\gamma^*})^qdt+( R_T L_T^{\gamma^*})^q\Big]}   \Big) ^{p-1} \;,
\end{equation}
and $\gamma^*$ is given by (\ref{optimalgamma}).
%The optimal strategy is
%\begin{eqnarray*}
%    \pi_t^* = \Big(\frac{\wh\varphi_t}{\Phi_{t^-}}+ \frac{\theta_t}{1-p}\Big)\frac{1}{\sigma_t}\;.
%\end{eqnarray*}
\end{Proposition}
%\begin{Remark}
%The previous proposition gives the optimal strategy in terms of the solution of the dual BSDE (\ref{BSDEPhi}). We can derive the optimal strategy in the particular case with stochastic interest rate without change of regime. The optimal strategy of this classical case is often characterized with the solution of the BSDE associated to the primal problem (see, e.g., Cheridito and Hu \cite{Cheridito}).
%\end{Remark}
%\begin{Remark}
%At the time of regime change, $\Phi, \wh\varphi,\tilde{\varphi}$ and $\theta$ jump and although the price of risky asset $S$ and the portfolio process $X$ do not jump at that time, the investor needs to rebalance the portfolio and change the consumption policy. From Remark \ref{Remark:Deterministic:Case}, we have that if $r^1, \theta^1, r^0, \theta^0$ and $\lambda^\ff$ are deterministic, the investors adjusts the portfolio at $\tau$ as follows
%\begin{eqnarray*}
%	\Delta \pi^*_{\tau} = \frac{r_{\tau-}-r_\tau}{(1-p)\sigma_\tau^2}\;.
%\end{eqnarray*}
%\end{Remark}
Before proving Proposition \ref{Prop:Existence:OptimalSolution}, we prove that the strategy $(\pi^*, c^*)$ is admissible.

%As $c_t^* = -(\tilde{U})^\prime(\eta^* R_tL_t^{\gamma^*})$ and $X_T^*=-(\tilde{U})^\prime(\eta^* R_TL_T^{\gamma^*})$, we obtain \eqref{optwealth} easily. In order to prove that $X^*$ is the optimal wealth process, it is sufficient to check that $X^*_T$ is duplicable.
%We demonstrate first the following lemma
\begin{Lemma}\label{Lemma:Admissible:Strategy}
	The strategy $(\pi^*, c^*)$ given by (\ref{optwealth}) is admissible and the wealth associated to  $(\pi^*, c^*)$ is
	\begin{eqnarray}\label{wealth:process:X}
		X^{x, \pi^*, c^*}_t= \big(  \eta^*R_t L_t^{\gamma^*} \big)^{\frac{1}{p-1}}\Phi_t\;.
	\end{eqnarray}
\end{Lemma}

\begin{proof}
Using  assumptions
${\bf (H3)}$ and ${\bf (H4)}$, and the properties of $(\Phi,  \widehat{\varphi} ,  \tilde{\varphi})$ given by Theorem \ref{theoreme existence}, we obtain that
%Since $\widehat{\varphi}$
%belonging to $\mathcal{H}^2_{\bG}(0,T)$,
%and there exists a constant $C$ such that $\Phi_t \geq C>0$, $\pi^*$ satisfies
$\bE (\int_0^T |\pi_s^*\sigma_s|^2 ds)<\infty$ and $\pi^*$ is $\bG$-predictable. Moreover, from (\ref{wealtheq}), the wealth process $X^{x, \pi^*, c^*}$ associated to the strategy $(\pi^*, c^*)$ is defined by the SDE
%\begin{eqnarray*}
%    dX^*_t = X^*_t\Big[ (r_t-(q-1)\theta_t^2+\theta_t\frac{\wh\varphi_t}{\Phi_{t^-}})dt + (\frac{\wh\varphi_t}{\Phi_{t^-}}-(q-1)\theta_t)dB_t\Big]-k\left(  R_t L_t^{\gamma^*} \right)^{\frac{1}{p-1}}dt
%\end{eqnarray*}
\begin{equation}\label{equation:sde:wealth}
\left\{\begin{array}{rcl}
X_0^{x, \pi^*, c^*} &= & x\;, \\
dX_t^{x, \pi^*, c^*} &=& X_t^{x, \pi^*, c^*}\Big[ \Big(r_t-(q-1)|\theta_t|^2+\theta_t\frac{\wh\varphi_t}{\Phi_{t^-}}\Big)dt + \Big(\frac{\wh\varphi_t}{\Phi_{t^-}}-(q-1)\theta_t\Big)dB_t\Big]\\
&&-\big(\eta^*  R_t L_{t^-}^{\gamma^*} \big)^{\frac{1}{p-1}}dt\;.
\end{array}\right.
\end{equation}
From Proposition \ref{prop:verification:martingale}, the process $\int_0^.( R_s L^{\gamma^*}_s)^q ds +  ( R L^{\gamma^*})^q \Phi\;$
%From the construction of $\{(J^{(d)}_t(\gamma))_{t \in [0,T]} ~: ~\gamma \in \Gamma\}$, we have $(\int_0^t (R_s L^{\gamma^*}_s)^q ds + (R_t L^{\gamma^*}_t)^q \Phi_t)_{t \in [0,T]}$
is a $\mathbb{G}$-martingale, which implies
\begin{equation*}
\Phi_0 = \bE\Big[\int_0^T(R_t L_t^{\gamma^*})^qdt+( R_T L_T^{\gamma^*})^q\Big]\;.
\end{equation*}
From the previous equality and (\ref{optlambda}), we remark that
$
(\eta^*)^{\frac{1}{p-1}}\Phi_0 = x\;.
$
Using It\^o's formula and (\ref{BSDEPhi0}), we check that $(\eta^*RL^{ \gamma^*})^{\frac{1}{p-1}}\Phi$ is a solution of the SDE (\ref{equation:sde:wealth}). Moreover, this SDE admits a unique solution. Therefore, we have
 \begin{eqnarray}\label{EqDuality21}
        X_t^{x, \pi^*, c^*} = (\eta^*R_tL_t^{ \gamma^*})^{\frac{1}{p-1}}\Phi_t\;.
    \end{eqnarray}
    %Since $\Phi_T = 1$, (\ref{wealth:process:X}) follows.
    Using the fact that $c_t^* \geq 0$ and $X_t^{x, \pi^*, c^*} >0$ for any $t \in [0,T]$, we conclude the proof. %we conclude that $(\pi^*, c^*)$ is admissible.
     In particular, $ (\eta^*R_TL_T^{ \gamma^*})^{\frac{1}{p-1}}= X_T^{x, \pi^*, c^*}$ is hedgeable.
\end{proof}
We now prove Proposition \ref{Prop:Existence:OptimalSolution}.
\begin{proof}
  	
From (\ref{eq:primal:dual:relation:1}), we obtain
\begin{eqnarray*}
	\sup_{(\pi,c) \in \A(x)}\bE \Big[ \int_0^T U(c_s) ds + U(X^{x, \pi, c}_T) \Big] \leq \inf_{\eta>0, \gamma \in \Gamma} \Big(-\frac{\eta^q}{q}\bE \Big[ \int_0^T (R_s L_s^{\gamma})^q ds + (R_T L_T^{\gamma})^q \Big] + \eta x\Big)\;.  %\\
\end{eqnarray*}	
By the definition of $\gamma^*$ and $\eta^*$, the previous inequality is equivalent to
\begin{eqnarray}	\label{eq:optimality:A}
	\sup_{(\pi,c) \in \A(x)}\bE \Big[ \int_0^T U(c_s) ds + U(X^{x, \pi, c}_T) \Big]	\leq
	%-\frac{(\eta^*)^q}{q}\bE \Big[ \int_0^T (R_s L_s^{\gamma^*})^q ds + (R_T L_T^{\gamma^*})^q \Big] + \eta^* x\\
	%&=&
	\frac{x^p}{p}\Big(\bE \Big[ \int_0^T (R_s L_s^{\gamma^*})^q ds + (R_T L_T^{\gamma^*})^q \Big]\Big)^{1-p}\;.
\end{eqnarray}
By definition of $(\pi^*, c^*)$ and Lemma \ref{Lemma:Admissible:Strategy}, we remark that
\begin{eqnarray}\label{eq:optimality:B}
	\bE \Big[ \int_0^T U(c_s^*) ds + U(X^{x, \pi^*, c^*}_T) \Big] = \frac{x^p}{p}\Big(\bE \Big[ \int_0^T (R_s L_s^{\gamma^*})^q ds + (R_T L_T^{\gamma^*})^q \Big]\Big)^{1-p}\;.
\end{eqnarray}
Since $(\pi^*, c^*)$ is admissible, from (\ref{eq:optimality:A}) and (\ref{eq:optimality:B}), we obtain

\begin{eqnarray*}%\label{Proof:Optimality}
	\bE \Big[ \int_0^T U(c_s^*) ds + U(X^{x, \pi^*, c^*}_T) \Big] = \sup_{(\pi,c) \in \A(x)}\bE \Big[ \int_0^T U(c_s) ds + U(X^{x, \pi, c}_T) \Big] \;.
\end{eqnarray*}
Therefore, $(\pi^*, c^*)$ is an optimal solution of the primal problem (\ref{pb}).
\end{proof}

We now characterize the value function associated to the primal problem using the dynamic programming principle.  For fixed $t \in [0, T]$  we denote by $(\pi^{t}, c^{t})$ a strategy defined on the time interval  $[t, T]$  and $(X_s^{t, x, \pi^{t}, c^{t}})_{s\in[t, T]}$ the wealth process associated to this strategy given the initial value at time $t$ is $x>0$. We first define  the set of control for a fixed $t\leq T$.

\begin{Definition}%\label{DefinitionAdmissibleStrategies}
The set $\A_t(x)$ of admissible strategies $(\pi^{t}, c^{t})$ from time $t$  consists in the set
of $\bG$-predictable processes $(\pi^{t}_s, c^{t}_s)_{s\in[t, T]}$ such that
$\bE (\int_t^T
|\pi^{t}_s\sigma_s|^2 ds)<\infty$, $c^{t}_s \geq 0$ and $X^{t, x,\pi^{t},c^{t}}_s  > 0$ for any $s \in [t,T]$.

\end{Definition}
%\begin{Definition}%\label{DefinitionAdmissibleStrategies}
%The set $\Gamma_t$ consists in the set of $\bG$-predictable processes $\gamma$ such
%that $\gamma_{t'} > -1$ for any $t'\in [t,T]$.
%\end{Definition}

%We define the value function at time
%$t \leq T$ for the primal problem as follows% (\textcolor{red}{Je pense que c'est mieux de d\'efinir $V$ et $\wt V$ avant, on ne sait pas ici non plus que signifie $V$ et $\wt V$}):
%\begin{eqnarray}\label{PrimalValueFunction5}
%    V(t, x)   := \frac{x^p}{p} {\Psi}_t (x)\;,
%%    V(t, X_t^{x, \pi, c})   &=& \int_0^t \frac{c_s^p}{p}ds+ \frac{(X_t^{x, \pi, c})^p}{p} {\Psi}_t\\
%%   \wt V (  t, \eta R_tL_t^{\gamma})    &=& -\frac{1}{q}\Big(\int_0^t (\eta R_sL_s^{\gamma})^qds+ (\eta R_tL_t^{\gamma})^q {\Phi}_t\Big)
%\end{eqnarray}
%where
%\begin{eqnarray*}
%     {\Psi}_t (x)   := \frac{1}{x^p}\esssup_{(\pi^{t}, c^{t})\in\mathcal{A}_t(x)}\bE\Big[\int_t^T (c_s^{t})^pds+(X_T^{t, x, \pi^{t}, c^{t}})^p\Big\vert\mathcal{G}_t\Big]\;.%\\&=&\sup_{\pi\in\mathcal{A}_t}E\Big[    \exp \left(\int_t^T (r_s  + \pi_s(\nu_s - r_s) ) ds \right) \;
%%     \exp \left( \int_t^T   \pi_s \sigma_s d B_s -\frac 12   \pi_s^2  \sigma_s ^2 ds\right) \Big]
%%  \\    {\Phi}_t    &=& \inf_{\gamma\in\Gamma_t}E\Big[\left(\frac{R_TL_T^{\gamma}}{R_tL_t^{\gamma}}\right)^q\vert\mathcal{G}_t\Big]
%\end{eqnarray*} %DE}FINIR $\mathcal{A}(t,x)$ et $\Gamma_t$.
%%According to Kramkov and Schachermayer \cite{KS},
We define the value function at time
$t \leq T$ for the primal problem as follows
\begin{eqnarray}\label{PrimalValueFunction5}
    V(t, x)   := \frac{x^p}{p} {\Psi}_t (x)\;,
\end{eqnarray}
where
\begin{eqnarray*}
     {\Psi}_t (x)   := \frac{1}{x^p}\esssup_{(\pi^{t}, c^{t})\in\mathcal{A}_t(x)}\bE\Big[\int_t^T (c_s^{t})^pds+(X_T^{t, x, \pi^{t}, c^{t}})^p\Big\vert\mathcal{G}_t\Big]\;.
\end{eqnarray*} 
For any $(\pi^{t}, c^{t})\in\mathcal{A}_t(x)$, we define the strategy $(\hat\pi^{t},  \hat c^{t})$ by $\hat\pi^{t} := \pi^{t}$ and $\hat c^{t} := c^{t}/x$.
We remark that $(\hat\pi^{t}, \hat c^{t})\in\mathcal{A}_t(1)$ and $X_T^{t, x, \pi^{t}, c^{t}} = x X_T^{t, 1, \hat\pi^{t}, \hat c^{t}}$ from (\ref{wealtheq}). Combining the previous relations with the definition of ${\Psi}_t (x)$, we obtain ${\Psi}_t (x) = {\Psi}_t (1)$. For the sake of brevity, we shall denote ${\Psi}_t$ instead of ${\Psi}_t (1)$. The value function at time $t\leq T$ can be rewritten as follows $V(t, x)   = x^p\Psi_t/p$.
From (\ref{PrimalValueFunction5}) and Proposition \ref{Prop:Existence:OptimalSolution}, we have
\begin{eqnarray}\label{Eq:Existence:OptimalSolution}
V(0, x) = \bE\Big[\int_0^T \frac{(c_s^*)^p}{p}ds+ \frac{(X_T^{x, \pi^*, c^*})^p}{p}\Big] = V(x)\;.
\end{eqnarray}
Using dynamic control techniques, we derive the following characterization of the value function. 
\begin{Proposition}\label{PropositionDynamicProgrammingDuale}
For any $(\pi, c)\in\A(x)$, $\int_0^. \frac{(c_s)^p}{p}ds+V(., X^{x, \pi, c})$ is a
$\mathbb{G}$-supermartingale
and there exists $(\pi^*, c^*)\in\A(x)$ such that $\int_0^. \frac{(c_s^*)^p}{p}ds+V(., X^{x, \pi^*, c^*})$ is a $\mathbb{G}$-martingale.
\end{Proposition}
\noindent The proof of this proposition is given in El Karoui \cite{EK}.\\

Using these properties, we can characterize the value function with a BSDE. %The process  $\Psi$ is defined by $\Psi_t = \Psi_t(X_t^{x, \pi^*, c^*})$ for any $t\in[0, T]$.
\begin{Proposition}\label{Prop:Primal:BSDE:Relation}
The process $\Psi$ satisfies the equality $\Psi= \Phi^{1-p}$. Moreover, the process $\Psi$ is  solution of the BSDE%and the following BSDE
 %\textcolor{red}{R\'e\'ecrire la condition terminale de $\Psi$}
%$$d\Psi_t=    \left(-pr_t \Psi_t+\frac 12 \frac {p}{p-1}( \theta^2 \Psi_t-\frac{\wh\psi_t^2}{\Psi_t})-
%p\lambda_t^\bG\Psi_t +\frac{p}{p-1} \theta_t \wh \psi_t \right)dt
%+\wh \psi_t dW_t+\wt \psi_t dM_t$$ \end{Proposition} \mb{VERIFIER,
%J AI DU OUBLIER UN $\wt \psi$ dans le drift}
%\textcolor{red}{
%\begin{equation} \label{BSDEPhi0}
%\begin{split}
%\Phi_t & = 1 - \int_t^T \Big( \big(q r_s - \frac{1}{2}q (q-1)  |\theta_s|^2  +
%(1-q) \lambda_s^\bG \big) \Phi_{s} + q \theta_s \widehat{\varphi}_s  + \lambda_s ^\bG\tilde{\varphi}_s \\
%& \quad  - (1-q) \lambda_s^\bG  (\Phi_{s} +
%\tilde{\varphi}_s)^{1-p} \;  \Phi_{s}^p -1 \Big) ds - \int_t^T
%\widehat{\varphi}_s d B_s - \int_t^T \tilde{\varphi}_s d M_s\;,
%\end{split}
%\end{equation}
\begin{equation}\label{Primal:BSDE:Relation}
\begin{split}
        \Psi_t &=   1-\int_t^T \Big(-(1-p)\Psi_s^q-pr_s   \Psi_s+\frac 12 \frac {p}{p-1}( |\theta_s|^2 \Psi_s+\frac{\wh\psi_s^2}{\Psi_s}) +\frac{p}{p-1} \theta_s \wh \psi_s   \Big)ds\\
&\quad-\int_t^T\wh \psi_s dB_s-\int_t^T\wt \psi_s dH_s\;.
\end{split}
\end{equation}
%}
%\textcolor{red}{CHECK}
%\textcolor{red}{\`a v\'erifier, je ne suis pas tr\`es s\^ur
%$$d\Psi_t=    \left(-pr_t \Psi_t+\frac 12 \frac {p}{p-1}( \theta^2 \Psi_t+\frac{\wh\psi_t^2}{\Psi_t})+((2-
%p)\Psi_t+2\wt \psi)\lambda_t^\bG +\frac{p}{p-1} \theta_t \wh \psi_t \right)dt
%+\wh \psi_t dW_t+\wt \psi_t dM_t$$}
%$$d\Psi_t=    \left(-pr_t \Psi_t+\frac 12 \frac {p}{p-1}( \theta^2 \Psi_t+\frac{\wh\psi_t^2}{\Psi_t}) +\frac{p}{p-1} \theta_t \wh \psi_t \right)dt
%+\wh \psi_t dB_t+\wt \psi_t dM_t$$}

%\nhn{$B$ or $W$ Brownian motion}
\end{Proposition}
\begin{proof}
From (\ref{optwealth}), (\ref{wealth:process:X}) and (\ref{PrimalValueFunction5}), we have
\begin{eqnarray*}
	\int_0^T \frac{(c_s^*)^p}{p}ds + V(T, X_T^{x, \pi^*, c^*})% &=& \int_0^T \frac{(c_s^*)p}{p}ds+ \frac{(X_T^{x, \pi^*, c^*})^p}{p}\\
					= \frac{(\eta^*)^q}{p}\Big( \int_0^T (R_sL_s^{\gamma^*})^qds + (R_TL_T^{\gamma^*})^q\Phi_T\Big)\;.%\\
					%= \frac{(\eta^*)^q}{p} R_T^{(d)}(\gamma^*)
\end{eqnarray*}
From Propositions  \ref{prop:verification:martingale} and \ref{PropositionDynamicProgrammingDuale}, the process $\int_0^. \frac{(c_s^*)^p}{p}ds+V(., X_.^{x, \pi^*, c^*})$ and $\int_0^.( R_s L^{\gamma^*}_s)^q ds +  ( R L^{\gamma^*})^q \Phi\;$ are $\mathbb{G}$-martingales. Therefore, taking the conditional expectation for the above equality, we obtain
\begin{eqnarray*}
	\int_0^t (c_s^*)^pds+ (X_t^{x, \pi^*, c^*})^p\Psi_t
					%&=& \frac{(\eta^*)^q}{p} R_t^{(d)}(\gamma^*)\\
					&=& (\eta^*)^q\Big( \int_0^t (R_sL_s^{\gamma^*})^qds + (R_tL_t^{\gamma^*})^q\Phi_t\Big)\;.
\end{eqnarray*}
Since $c_t^* =  (  \eta^*R_t L_t^{\gamma^*} )^{\frac{1}{p-1}}$, the following relation
    holds for any $t\in[0, T]$
    \begin{eqnarray}\label{EqDuality11bis}
        (X^{x, \pi^*, c^*}_t)^p {\Psi_t}=(\eta^*)^q( R_tL_t^{ \gamma*})^q\Phi_t\;.
    \end{eqnarray}
    Therefore, from (\ref{EqDuality21}) and (\ref{EqDuality11bis}), we obtain
    \begin{eqnarray*}
        \Phi_t^{1-p} =  {\Psi_t}\;.
    \end{eqnarray*}
  Applying It\^o's formula to $\Phi^{1-p}$, we obtain
\begin{equation*} %\label{BSDEPhi0}
\begin{split}
d\Phi_t^{1-p} & = (1-p) \Phi_{t^-}^{-p}\Big( \big(q r_t - \frac{1}{2}q (q-1)  |\theta_t|^2  + \lambda^\bG_t (1-q)\big) \Phi_{t} + q \theta_t \widehat{\varphi}_t     -(1-q)\lambda^\bG_t (\Phi_{t^-} + \tilde \varphi_t)^{1-p} \Phi_t^p\\
& \quad -1 -\frac{1}{2}p\Phi_t^{-1}\widehat{\varphi}_t^2 \Big) dt
 + (1-p)
\wh \varphi_t \Phi_{t^-}^{-p} d B_t + \Big((\Phi_{t^-}+\wt
\varphi_t)^{1-p}-\Phi_{t^-}^{1-p}\Big) d H_t\;.
\end{split}
\end{equation*}
  Setting $\wh \psi=(1-p)
\wh \varphi \Phi^{-p}$ and $ \wt \psi =(\Phi+\wt
\varphi)^{1-p}-\Phi^{1-p}\;,$ we get that $(\Psi, \wh \psi, \wt \psi)$ satisfies (\ref{Primal:BSDE:Relation}).
\end{proof}

The above result is not sufficient to characterize the value function of the primal problem since it is not obvious that the BSDE \reff{Primal:BSDE:Relation} admits a unique solution. But thanks to the uniqueness of the solution of the BSDE \reff{BSDEPhi0} we get the following characterization of the value function of the primal problem.
\begin{Theorem}\label{theoreme unicite}
$\Psi$ is the unique solution of the BSDE \reff{Primal:BSDE:Relation} in $\mathcal{S}^{\infty,+}_{\bG}(0, T) \times
\mathcal{H}^2_{\bG}(0,T)   \times
 \mathcal{H}^2_{\bG}(M)$.
\end{Theorem}

\begin{proof}
Let $(y,z,u) \in  \mathcal{S}^{\infty,+}_{\bG}(0, T) \times
\mathcal{H}^2_{\bG}(0,T)   \times
 \mathcal{H}^2_{\bG}(M)$ be a solution of the BSDE \reff{Primal:BSDE:Relation}. We define $Y_t:= y^{1-q}_t$, $Z_t:= (1-q) y_t^{-q} z_t$ and $U_t:=(y_{t^-} + u_t)^{1-q} - y^{1-q}_{t^-}$ for any $t \in [0,T]$. From It\^o's formula, we get that 
\begin{eqnarray*}
 dY_t &= &\Big[   \big( qr_t -\frac{q(q-1)}{2} |\theta_t|^2 + (1-q)\lambda^\bG_t \big) Y_t  - (1-q)\lambda^\bG_t  (U_t + Y_t)^{1-p}Y_t^p\\&& + q \theta_t Z_t - 1\Big]dt+ Z_t dB_t + U_t dH_t\;.\end{eqnarray*}
Therefore, $(Y,Z,U)$ is solution of the BSDE \reff{BSDEPhi0} and, from Subsection \ref{sous section unite}, we have by uniqueness of the solution $Y=\Phi$ and from Proposition \ref{Prop:Primal:BSDE:Relation} we get that $y=\Psi$.
\end{proof}

\begin{Remark}
From Theorems \ref{theoreme existence} and \ref{theoreme unicite} and Proposition \ref{Prop:Primal:BSDE:Relation}, we can conclude that the BSDE (\ref{Primal:BSDE:Relation}), which is associated to the primal problem (\ref{pb}), admits a unique solution $(\Psi, \wh \psi, \wt \psi)$ belonging to $\mathcal{S}^{\infty, +}_{\bG}(0, T) \times
\mathcal{H}^2_{\bG}(0,T)   \times
 \mathcal{H}^2_{\bG}(M)$. Solving this BSDE directly is not evident because of the terms $\Psi^q$ and $\frac{\wh\psi^2}{\Psi}$.
\end{Remark}

\begin{Remark} We point out that, in the case where the coefficients of the model are deterministic functions  of some external economic factors and in a Brownian setting, the optimal control processes $(\pi^*, c^*)$ have the same expressions that  those  obtained by Casta\~{n}eda-Leyva and Hern\'{a}ndez-Hern\'{a}ndez \cite{CH} (see Proposition 3.1. in \cite{CH}). \\
We also remark that,  in a default-density setting,  the optimal control processes $(\pi^*, c^*)$ have the same expressions that  those  obtained by  Jiao and Pham \cite{JP}.  In these two papers, the optimal portfolio is given in terms of the value function of the primal problem or  in terms of the solution of the primal BSDE as $$\pi^*_t = \frac{1}{(1-p)\sigma_t}\Big(\frac{\widehat \psi_t}{\Psi_{t^-}} +\theta_t\Big)\;.$$ We have proved that $\frac{1}{1-p}\frac{\widehat \psi_t}{\Psi_{t^-}}= \frac{\widehat \varphi_t}{ \Phi_ {t^-}}$, hence, the two solutions have the same form.
In particular, after $\tau$, our setting is that one of a complete market, and our formula is rather standard. In particular, in the case where $r^1$ and $\theta^1$ are deterministic, the investor is myopic, and  the optimal portfolio is  $\pi^* = \frac{1}{1-p}  \frac{\theta}{\sigma}$. Before $\tau$, the investor takes into account the fact that the interest rate will change, since the after default value function appears in the before default value function (this is the term $ \Phi^1_t(t)$ in the associated BSDE). \end{Remark}

\section{Conclusion}
In this paper, we have studied  the problem of maximization of expected power utility of both terminal wealth and consumption  in a market with a stochastic interest rate in a model where immersion holds. We have  derived the optimal strategy solving    the associated dual problem. Then, we have given  the link between the value functions associated to the primal and dual problems, which has allowed to characterize the value function of the primal problem by a BSDE.

 If one assumes that $B$ is a $\bG$-semi-martingale with  canonical decomposition of $B$ in $\bG$  of the  form
\begin{eqnarray*}
B_t = B_t^\bG + \int_0^t\mu_sds\;,
\end{eqnarray*}
with a bounded process $\mu$ and $B^\bG$ a $\bG$-Brownian motion, the price dynamics of the risky asset can be rewritten as follows
\begin{equation*}
d S_t = S_t \Big((\nu_t+\sigma_t\mu_t) dt  + \sigma_t d B_t^\bG\Big) \;,
\end{equation*}
where the coefficients $\nu$ and $\sigma$  can be chosen   $\bG$-predictable and bounded.
In this case, the e.m.m.  can be written on the form (indeed, a predictable representation property holds  for the pair $B^\bG,M$)
$$dL_t=  L_{t^-}(a_tdB^\bG _t+\gamma_t dM_t)$$ and
using the same methods and arguments, we can obtain similar results. The real difficulty is that one has to assume that the process $\mu$ is bounded, and we do not know any condition on $\tau$ which implies that fact.

Without any theoretical difficulty, we can generalize this paper to the case where there are several  ordered changes of regime of interest rate.

\end{document}